\newcommand{\PP}{Perkowski/Promel:2013}

\documentclass{article}
\usepackage{amsmath,amsthm,amsfonts,amssymb,latexsym,euscript,bm,url}

\allowdisplaybreaks[4]

\DeclareMathOperator{\III}{\boldsymbol{1}}
\DeclareMathOperator{\sign}{sign}

\DeclareMathOperator{\UpProb}{\overline{\mathbb{P}}}

\DeclareMathOperator{\Expect}{\mathbb{E}}

\DeclareMathOperator{\var}{var}

\DeclareMathOperator{\vi}{vi}

\newcommand{\K}{\EuScript{K}}

\newcommand{\st}{\mid}

\newtheorem{theorem}{Theorem}
\newtheorem{lemma}{Lemma}

\theoremstyle{definition}
\newtheorem{remark}{Remark}

\newcommand*{\CTIV}{Vovk:arXiv0904}
\newcommand*{\CTV}{Vovk:arXiv1005}

\title{Getting rich quick with the Axiom of Choice}
\author{Vladimir Vovk}

\begin{document}
\maketitle

  \begin{abstract}
    This paper proposes new get-rich-quick schemes that involve trading in a financial security
    with a non-degenerate price path.
    For simplicity the interest rate is assumed zero.
    If the price path is assumed continuous, the trader can become infinitely rich
    immediately after it becomes non-constant (if it ever does).
    If it is assumed positive, he can become infinitely rich immediately after reaching a point in time
    such that the variation of the log price is infinite in any right neighbourhood of that point
    (whereas reaching a point in time such that the variation of the log price is infinite in any left neighbourhood of that point
    is not sufficient).
    The practical value of these schemes is tempered by their use of the Axiom of Choice.

    \bigskip

    \noindent
    The version of this paper at \url{http://probabilityandfinance.com} (Working Paper 43)
    is updated most often.
    The journal version is to appear in \emph{Finance and Stochastics}
    under the title ``The role of measurability in game-theoretic probability''.
  \end{abstract}

\section{Introduction}

This paper belongs to the area of game-theoretic probability
(see, e.g., \cite{Shafer/Vovk:2001}).
The advantage of game-theoretic probability for mathematical finance
over the dominant approach of measure-theoretic probability
is that it allows us to state and prove results free of any statistical assumptions
even in situations where such assumptions are often regarded as essential
(see, e.g., the probability-free Dubins--Schwarz theorem in \cite{\CTIV}
and the probability-free theory of stochastic integration in \cite{\PP}).
In this paper we consider the framework of an idealized financial market
with one tradable security and assume, for simplicity, a zero interest rate.

The necessity of a non-trivial requirement of measurability (such as Borel, Lebesgue, or universal)
is well known in measure theory
(without measurability we have counter-intuitive results such as the Banach--Tarski paradox\footnote{%
  The Banach--Tarski paradox is sometimes also regarded as enabling get-rich-quick schemes, such as:
  (a) buy a golden ball; (b) duplicate your ball; (c) sell the extra ball; (d) go to (b).
  This scheme, however, is not as quick as ours.}
\cite{Wagon:1985}),
and it is inherited by measure-theoretic probability.
In game-theoretic probability measurability is usually not needed in discrete time
(and is never assumed in, e.g., \cite{Shafer/Vovk:2001});
this paper, however, shows that, in continuous time,
imposing some regularity conditions (such as Borel or universal measurability) is essential
even in the foundations of game-theoretic probability.
If such conditions are not imposed, the basic definitions of game-theoretic probability
become uninteresting, or even degenerate:
e.g., in the case of continuous price paths the upper probability of sets
can take only two values:
1 (if the set contains a constant price path) or 0 (if not).

This paper constructs explicit trading strategies for enriching the trader
given a well-order of the space of all possible price paths;
however, such a well-order exists only under the Axiom of Choice
(which is, despite some anomalous corollaries, universally accepted).
Since we cannot construct such well-orders,
our strategies cannot be regarded as genuinely constructive.
Therefore, they cannot be regarded as practical get-rich-quick schemes.
Moreover, they do not affect the existing results of continuous-time game-theoretic probability,
which always explicitly assume measurability (to the best of my knowledge).

Our trading strategies will be very simple and based on Hardin and Taylor's work on hat puzzles
  (going back to at least \cite{Gardner:1959}
  and \cite[Chapter~4]{Gamow/Stern:1958}).
These authors show in \cite{Hardin/Taylor:2008} (see also \cite[Section~7.4]{Hardin/Taylor:2013})
that the Axiom of Choice provides us with an Ockham-type strategy
able to predict short-term future (usually albeit not always),
which makes it easy to get rich when allowed to trade in a security
whose price changes in a non-trivial manner.

Section~\ref{sec:continuous} is devoted to continuous price paths.
The assumption of continuity allows us to use leverage and stop-loss strategies,
and the trader can profit greatly and quickly whenever the price path is not constant.
(We only consider trading strategies that never risk bankruptcy.
It is clear that profiting from a constant price path is impossible.)

In Section~\ref{sec:cadlag} we assume, instead of continuity, that the price path is c\`adl\`ag.
To make trading possible we further assume that the price path is positive.
It is impossible for the trader to become infinitely rich if the log price path has a finite variation.
If the variation is infinite,
there will be either points in time such that the variation of the log price
is infinite in any of their left neighbourhoods
or points in time such that the variation of the log price is infinite
in any of their right neighbourhoods.
Becoming infinitely rich is possible after points of the latter type.
Standard stochastic models of financial markets postulate price paths that have such points almost surely.

In the short Section~\ref{sec:positive} we only assume that the price path is positive;
the theory in this case is almost identical to the theory for positive c\`adl\`ag price paths.

The proofs of all our main results are collected in a separate section,
Section~\ref{sec:proofs};
they are based on Hardin and Taylor's results.
Appendix~A provides a more general picture of predicting short-term future
using the Axiom of Choice.
It answers several very natural questions
(and at the end asks 256 more questions answering just one of them).

Our definitions of the basic notions of continuous-time probability
(such as stopping times)
will be Galmarino-type (see, e.g., \cite[Theorems~1.2 and~1.4]{Courrege/Priouret:1965})
and modelled on the ones in the technical report \cite{\CTIV}
(the journal version uses slightly different definitions),
except that the requirements of measurability will be dropped.
By ``positive'' I will mean ``nonnegative'',
adding ``strictly'' when necessary.
The restriction $f|_C$ of a function $f:A\to B$ to a set $C$ is defined as $f|_{A\cap C}$;
this notation will be used even when $C\not\subseteq A$.

\section{Continuous price paths}
\label{sec:continuous}

Let $\Omega$ be the set $C[0,1]$ of all continuous functions $\omega:[0,1]\to\mathbb{R}$
(intuitively, these are the potential price paths over the time interval $[0,1]$).
An \emph{adapted process} $\mathfrak{S}$ is a family of functions
$\mathfrak{S}_t:\Omega\to[-\infty,\infty]$, $t\in[0,1]$,
such that, for all $\omega,\omega'\in\Omega$ and all $t\in[0,1]$,
\begin{equation*}
  \omega|_{[0,t]}
  =
  \omega'|_{[0,t]}
  \Longrightarrow
  \mathfrak{S}_t(\omega)=\mathfrak{S}_t(\omega').
\end{equation*}
The intuition is that $\mathfrak{S}_t(\omega)$ depends on $\omega$ only via $\omega|_{[0,t]}$:
if $\omega$ changes over $(t,1]$, $\mathfrak{S}_t(\omega)$ is not affected.
A \emph{stopping time} is a function
$\tau:\Omega\to[0,1]$
such that, for all $\omega,\omega'\in\Omega$,
\begin{equation}\label{eq:tau}
  \omega|_{[0,\tau(\omega)]}
  =
  \omega'|_{[0,\tau(\omega)]}
  \Longrightarrow
  \tau(\omega)=\tau(\omega').
\end{equation}
The intuition is that $\tau(\omega)$ is not affected
if $\omega$ changes over $(\tau(\omega),1]$.
For any stopping time $\tau$,
a function $X:\Omega\to\mathbb{R}$ is said to be
\emph{determined by} time $\tau$ if,
for all $\omega,\omega'\in\Omega$,
\begin{equation}\label{eq:X}
  \omega|_{[0,\tau(\omega)]}
  =
  \omega'|_{[0,\tau(\omega)]}
  \Longrightarrow
  X(\omega)=X(\omega').
\end{equation}
The intuition is that $X(\omega)$ depends on $\omega$ only via $\omega|_{[0,\tau(\omega)]}$.
We will often simplify $\omega(\tau(\omega))$ to $\omega(\tau)$
(occasionally, the argument $\omega$ will be omitted in other cases as well).

The class of allowed trading strategies is defined in two steps.
First, a \emph{simple trading strategy} $G$
consists of an increasing sequence of stopping times
$\tau_1\le\tau_2\le\cdots$
and, for each $k=1,2,\ldots$,
a bounded function $h_k$ that is determined by time $\tau_k$.
It is required that, for each $\omega\in\Omega$,
$\tau_k(\omega)=1$ from some $k$ on.
To such $G$ and an \emph{initial capital} $c\in\mathbb{R}$
corresponds the \emph{simple capital process}
\begin{equation}\label{eq:simple-capital}
  \K^{G,c}_t(\omega)
  :=
  c
  +
  \sum_{k=1}^{\infty}
  h_k(\omega)
  \bigl(
    \omega(\tau_{k+1}\wedge t)-\omega(\tau_k\wedge t)
  \bigr),
  \quad
  t\in[0,1]
\end{equation}
(with the zero terms in the sum ignored,
which makes the sum finite);
the value $h_k(\omega)$ will be called the \emph{bet}
at time $\tau_k(\omega)$,
and $\K^{G,c}_t(\omega)$ will be referred to
as the capital at time $t$.

Second, a \emph{positive capital process} is any adapted process $\mathfrak{S}$
that can be represented in the form
\begin{equation}\label{eq:positive-capital}
  \mathfrak{S}_t(\omega)
  :=
  \sum_{n=1}^{\infty}
  \K^{G_n,c_n}_t(\omega),
\end{equation}
where the simple capital processes $\K^{G_n,c_n}_t(\omega)$
are required to be positive, for all $t$ and $\omega$,
and the positive series $\sum_{n=1}^{\infty}c_n$ is required to converge in $\mathbb{R}$.
The sum \eqref{eq:positive-capital} is always positive but allowed to take the value $\infty$.
Since $\K^{G_n,c_n}_0(\omega)=c_n$ does not depend on $\omega$,
$\mathfrak{S}_0(\omega)$ also does not depend on $\omega$
and will sometimes be abbreviated to $\mathfrak{S}_0$.

The \emph{upper probability} of a set $E\subseteq\Omega$ is defined as
\begin{equation}\label{eq:upper-probability}
  \UpProb(E)
  :=
  \inf
  \bigl\{
    \mathfrak{S}_0
    \bigm|
    \forall\omega\in\Omega:
    \mathfrak{S}_1(\omega)
    \ge
    \III_E(\omega)
  \bigr\},
\end{equation}
where $\mathfrak{S}$ ranges over the positive capital processes
and $\III_E$ stands for the indicator function of $E$.
We say that a set $E\subseteq\Omega$ is \emph{null} if $\UpProb(E)=0$.

\begin{remark}
  The intuition behind a simple trading strategy
  is that the trader is allowed to take positions $h_k$,
  either long or short,
  in a security whose price at time $t\in[0,1]$ is denoted $\omega(t)$.
  The positions can change only at a discrete sequence of times $\tau_1,\tau_2,\ldots$,
  which makes the definition \eqref{eq:simple-capital}
  of the trader's capital at time $t$ uncontroversial.
  To obtain more useful trading strategies,
  we allow the trader to split his initial capital into a countable number of accounts
  and to run a separate simple trading strategy for each account;
  none of the component simple trading strategies is allowed to go into debt.
  The resulting total capital at time $t$ is given by \eqref{eq:positive-capital}.
  The upper probability $\UpProb(E)$, defined by \eqref{eq:upper-probability},
  is the smallest initial capital sufficient for superhedging the binary option on $E$.
\end{remark}

The following theorem will be proved in Section~\ref{sec:proofs}.

\begin{theorem}\label{thm:continuous}
  The set of all non-constant $\omega\in\Omega$ is null.
  Moreover, there is a positive capital process $\mathfrak{S}$ with $\mathfrak{S}_0=1$ that becomes infinite
  as soon as $\omega$ ceases to be constant:
  for all $t\in[0,1]$,
  \begin{equation}\label{eq:continuous}
    \left(
       \exists t_1,t_2\in[0,t):
       \omega(t_1)\ne\omega(t_2)
    \right)
    \Longrightarrow
    \mathfrak{S}_t(\omega)=\infty.
  \end{equation}
\end{theorem}

\begin{remark}
  A more popular version of our definition \eqref{eq:upper-probability}
  was given by Perkowski and Pr\"omel \cite{\PP}.
  Perkowski and Pr\"omel's definition is more permissive \cite[Section~2.3]{\PP},
  and so the first statement of Theorem~\ref{thm:continuous} continues to hold for it as well
  if we allow non-measurable (but still non-anticipative) trading strategies.
  As all papers (that I am aware of) on continuous-time game-theoretic probability,
  the definitions given in \cite{\PP} assume the measurability of all strategies.
\end{remark}

\begin{remark}
  The definitions of this section assume that the trader is permitted to short the security
  (which allows $h_k(\omega)<0$)
  and borrow money (which allows leverage,
  $H_k(\omega)>1$ in the notation of Remark~\ref{rem:predictable} below).
  If shorting and borrowing are not permitted
  (in the notation of Remark~\ref{rem:predictable},
  if $H_k$ are only permitted to take values in $[0,1]$),
  Theorem~\ref{thm:continuous} ceases to be true,
  but Theorem~\ref{thm:cadlag-plus} is still applicable.
\end{remark}

\section{Positive c\`adl\`ag price paths}
\label{sec:cadlag}

In this section we will prove an analogue of Theorem~\ref{thm:continuous}
for positive c\`adl\`ag price paths $\omega$;
the picture now becomes more complicated.
Intuitively, $\omega:[0,1]\to[0,\infty)$ is a price path of a financial security
whose price is known always to stay positive
(such as stock, and from now on it will be referred to as stock).
For simplicity in the bulk of this section we consider the price paths $\omega$
satisfying $\inf\omega>0$;
the case of general positive $\omega$ will be considered in Remark~\ref{rem:cadlag}
at the end of the section.
Therefore, we redefine $\Omega$ as the set of all $\omega:[0,1]\to[0,\infty)$
such that $\inf\omega>0$.
The definitions of adapted processes, stopping times, etc., stay literally as before
(but with the new definition of $\Omega$).
Our goal will be to determine the sign of $\UpProb(E)$
(i.e., to determine whether $\UpProb(E)>0$)
for a wide family of sets $E\subseteq\Omega$.

For each function $f:[a,b]\to\mathbb{R}$, where $[a,b]\subseteq[0,1]$, its \emph{variation} $\var f$ is defined as
\begin{equation}\label{eq:variation}
  \var(f)
  :=
  \sup
  \sum_{i=1}^n
  \left|
    f(t_i) - f(t_{i-1})
  \right|
  \in
  [0,\infty],
\end{equation}
where the $\sup$ is taken over all $n=1,2,\ldots$ and all partitions
\[
  a=t_0<t_1<\cdots<t_n=b;
\]
this definition will usually be used for $[a,b]=[0,1]$.
As noticed in \cite{\CTV},
$\UpProb(\{\omega\})>0$ for each price path $\omega\in\Omega$ with $\var(\log\omega)<\infty$.
Namely, we have the following simple result
(to be proved in Subsection~\ref{subsec:thm-cadlag}).

\begin{theorem}\label{thm:cadlag}
  For any $\omega\in\Omega$,
  \begin{equation}\label{eq:bounded-variation}
    \UpProb(\{\omega\})
    =
    \sqrt{\frac{\omega(0)}{\omega(1)}e^{-\var(\log\omega)}}.
  \end{equation}
\end{theorem}

We can see that $\UpProb(E)>0$ whenever $E$ contains $\omega$ with $\var(\log\omega)<\infty$.
Therefore, in the rest of this section
we will concentrate on $\omega\in\Omega$ with $\var(\log\omega)=\infty$.
We start from a classification of such $\omega$.

For any $f:[0,1]\to\mathbb{R}$ and $t\in[0,1]$, set
\begin{align}
  \var(f,t-)
  &:=
  \inf_{t'\in[0,t)}
  \var
  \left(
    f|_{[t',t]}
  \right),
  \label{eq:var-}\\
  \var(f,t+)
  &:=
  \inf_{t'\in(t,1]}
  \var
  \left(
    f|_{[t,t']}
  \right);
  \label{eq:var+}
\end{align}
the cases $\var(f,0-):=0$ and $\var(f,1+):=0$ are treated separately.
Notice that $\var(f,t-)$ and $\var(f,t+)$ always take values
in the two-element sets $\{\lvert\Delta f(t)\rvert,\infty\}$ and $\{0,\infty\}$, respectively
(where $\Delta f(t):=f(t)-f(t-)$ is the jump of $f$ at $t$).
Furthermore, for each $\omega\in\Omega$ set
\begin{align}
  I^-_{\omega}
  &:=
  \left\{
    t\in[0,1]
    \st
    \var(\log\omega,t-) = \infty
  \right\}
  \subseteq
  (0,1],
  \label{eq:I-}\\
  I^+_{\omega}
  &:=
  \left\{
    t\in[0,1]
    \st
    \var(\log\omega,t+) = \infty
  \right\}
  \subseteq
  [0,1).
  \label{eq:I+}
\end{align}

The next lemma shows that the set of all $\omega\in\Omega$
with infinite variation of their logarithm
can be represented as
\begin{equation}\label{eq:union}
  \{\omega\in\Omega\st\var(\log\omega)=\infty\}
  =
  \{\omega\in\Omega\st I^-_{\omega}\ne\emptyset\}
  \cup
  \{\omega\in\Omega\st I^+_{\omega}\ne\emptyset\}.
\end{equation}

\begin{lemma}\label{lem:var}
  For any $f:[0,1]\to\mathbb{R}$,
  \begin{equation*}
    \var(f)=\infty
    \Longleftrightarrow
    \left(
      \exists t\in[0,1]:
      \var(f,t+)=\infty
      \vee
      \var(f,t-)=\infty
    \right).
  \end{equation*}
\end{lemma}

\begin{proof}
  The implication $\Longleftarrow$ is obvious,
  and so we only check $\Longrightarrow$.
  Suppose $\var(f,t+)<\infty$ and $\var(f,t-)<\infty$ for all $t\in[0,1]$.
  Fix a neighbourhood $O_t=(a_t,b_t)$ of each $t$ such that $\var(f|_{O_t})<\infty$.
  These neighbourhoods form a cover of $[0,1]$.
  The existence of its finite subcover
  immediately implies that $\var(f)<\infty$.
\end{proof}

The following theorem (proved in Section~\ref{sec:proofs}) tackles the second term of the union in~\eqref{eq:union}.

\begin{theorem}\label{thm:cadlag-plus}
  The set of all $\omega\in\Omega$ such that $I^+_{\omega}\ne\emptyset$ is null.
  Moreover, there is a positive capital process $\mathfrak{S}$ with $\mathfrak{S}_0=1$ that becomes infinite
  immediately after the time $\inf I^+_{\omega}$ if $I^+_{\omega}\ne\emptyset$:
  for all $t\in[0,1]$ and $\omega\in\Omega$,
  \begin{equation}\label{eq:cadlag-plus}
    \left(
      \exists t'\in[0,t):
      \var(\log\omega,t'+)=\infty
    \right)
    \Longrightarrow
    \mathfrak{S}_t(\omega)=\infty.
  \end{equation}
\end{theorem}

\begin{remark}\label{rem:predictable}
  A stopping time $\tau$ is said to be \emph{predictable}
  if \eqref{eq:tau} holds with the two entries of $[0,\tau(\omega)]$ replaced by $[0,\tau(\omega))$
  (cf.\ \cite[Theorem~IV.99(a)]{Dellacherie/Meyer:1975}).
  Similarly, a function $X:\Omega\to\mathbb{R}$ is \emph{determined before} a stopping time $\tau$
  if \eqref{eq:X} holds with the two entries of $[0,\tau(\omega)]$ replaced by $[0,\tau(\omega))$
  (cf.\ \cite[Theorem~IV.99(b)]{Dellacherie/Meyer:1975}).
  We will impose this requirement on the \emph{relative bets}
  \[
    H_k(\omega)
    :=
    h_k(\omega)
    \omega(\tau_k)
    /
    \K^{G,c}_{\tau_k}(\omega)
  \]
  involved in a simple capital process~\eqref{eq:simple-capital};
  intuitively, $H_k(\omega)$ is the fraction of the trader's capital
  invested in the stock at time $\tau_k$.
  In terms of the relative bets,
  the simple capital process~\eqref{eq:simple-capital} can be rewritten as
  \begin{equation*}
    \K^{G,c}_t(\omega)
    =
    c
    \prod_{k=1}^{\infty}
    \left(
      1
      +
      H_k(\omega)
      \left(
        \frac{\omega(\tau_{k+1}\wedge t)}{\omega(\tau_k\wedge t)}
        -
        1
      \right)
    \right),
    \quad
    t\in[0,1].
  \end{equation*}
  Notice that this simple capital process is positive if and only if the relative bets $H_k$
  are always in the range $[0,1]$
  (since the stock price can shoot up or drop nearly to $0$ at any time).
  A positive capital process \eqref{eq:positive-capital} is \emph{predictable}
  if the component simple capital processes $\K^{G_n,c_n}$
  involve only predictable stopping times $\tau_k$ and relative bets $H_k$ determined before $\tau_k$.
  Theorem~\ref{thm:cadlag-plus} can be strengthened by requiring the positive capital process $\mathfrak{S}$
  to be, in addition, predictable.
  (Notice that predictability was automatic
  in the continuous case of Section~\ref{sec:continuous}.)
  This observation can be strengthened further.
  Let us say that a stopping time $\tau$ is \emph{strongly predictable}
  if, for every $\omega\in\Omega$, there exists $t<\tau(\omega)$ such that,
  for every $\omega'\in\Omega$,
  \begin{equation*}
    \omega|_{[0,t]}
    =
    \omega'|_{[0,t]}
    \Longrightarrow
    \tau(\omega)=\tau(\omega').
  \end{equation*}
  A function $X:\Omega\to\mathbb{R}$ is said to be \emph{determined strictly before} a stopping time $\tau$ if,
  for every $\omega\in\Omega$,
  there exists $t<\tau(\omega)$ such that,
  for every $\omega'\in\Omega$,
  \begin{equation*}
    \omega|_{[0,t]}
    =
    \omega'|_{[0,t]}
    \Longrightarrow
    X(\omega) = X(\omega').
  \end{equation*}
  A positive capital process \eqref{eq:positive-capital} is \emph{strongly predictable}
  if the component simple capital processes $\K^{G_n,c_n}$
  involve only strongly predictable stopping times $\tau_k$
  and relative bets $H_k$ determined strictly before $\tau_k$.
  Theorem~\ref{thm:cadlag-plus} can be further strengthened
  by requiring the positive capital process $\mathfrak{S}$
  to be strongly predictable.
  A simple modification of the proof of Theorem~\ref{thm:cadlag-plus}
  demonstrating this fact will be given in Remark~\ref{rem:predictable-proof}.
\end{remark}

\begin{remark}
  In the context of the previous remark,
  imposing the requirement of being determined before $\tau_k$
  on the bets $h_k$ rather than relative bets $H_k$
  would lead to a useless notion of a predictable positive capital process:
  all such processes would be constant.
  In the continuous case of Section~\ref{sec:continuous}
  (where $\omega$ is not required to be positive),
  the notion of a predictable positive capital process
  is equivalent to that of a positive capital process,
  but the notion of a strongly predictable positive capital process,
  even as given in the previous remark,
  is useless:
  again, any such process is a constant.
\end{remark}

Theorems~\ref{thm:cadlag} and~\ref{thm:cadlag-plus} show that the only non-trivial part of $\Omega$
(as far as the sign of $\UpProb$ is concerned) is
\begin{equation*}
  \Omega^{\rm nt}
  :=
  \left\{
    \omega\in\Omega
    \st
    I^-_{\omega} \ne \emptyset,
    I^+_{\omega} = \emptyset
  \right\}.
\end{equation*}
Namely,
\begin{equation*}
  \UpProb(E)
  \begin{cases}
    {}> 0 & \text{if $\exists\omega\in E:\var(\log\omega)<\infty$}\\
    {}= \UpProb(E\cap\Omega^{\rm nt}) & \text{otherwise}.
  \end{cases}
\end{equation*}
Theorem~\ref{thm:cadlag} and the following result (also to be proved in Section~\ref{sec:proofs})
show that this part is really non-trivial:
it has subsets of upper probability one and non-empty subsets (such as any singleton) of upper probability zero.
\begin{theorem}\label{thm:cadlag-minus}
  The set $\Omega^{\rm nt}$ has upper probability one.
  Moreover,
  \begin{equation}\label{eq:cadlag-minus}
    \UpProb
    \left\{
      \omega\in\Omega
      \st
      I^-_{\omega}=\{t\},
      I^+_{\omega}=\emptyset
    \right\}
    =
    1
  \end{equation}
  for each $t\in(0,1]$.
\end{theorem}

\begin{remark}\label{rem:var-plus}
  The following modification of variation~\eqref{eq:variation} is often useful:
  \begin{equation}\label{eq:var-plus}
    \var^+(f)
    :=
    \sup
    \sum_{i=1}^n
    \left(
      f(t_i) - f(t_{i-1})
    \right)^+,
  \end{equation}
  where $u^+:=u\vee0$;
  we will allow $f:[0,1]\to[-\infty,\infty)$.
  Using this definition, \eqref{eq:bounded-variation} can be simplified
  (cf.\ \cite[the end of Section~2]{\CTV})
  to
  \begin{equation}\label{eq:e-var}
    \UpProb(\{\omega\})
    =
    e^{-\var^+(\log\omega)},
  \end{equation}
  and in this form the equality becomes true for any positive c\`adl\`ag $\omega$
  (with $\inf\omega=0$ allowed).
  The modified versions of \eqref{eq:var-}--\eqref{eq:I+} are:
  \begin{align}
    \var^+(f,t-)
    &:=
    \inf_{t'\in[0,t)}
    \var^+
    \left(
      f|_{[t',t]}
    \right),
    \label{eq:var-mod}\\
    \var^+(f,t+)
    &:=
    \inf_{t'\in(t,1]}
    \var^+
    \left(
      f|_{[t,t']}
    \right),
    \label{eq:var+mod}\\
    J^-_{\omega}
    &:=
    \left\{
      t\in[0,1]
      \st
      \var^+(\log\omega,t-) = \infty
    \right\}
    \subseteq
    (0,1],
    \label{eq:J-}\\
    J^+_{\omega}
    &:=
    \left\{
      t\in[0,1]
      \st
      \var^+(\log\omega,t+) = \infty
    \right\}
    \subseteq
    [0,1).
    \label{eq:J+}
  \end{align}
  Lemma~\ref{lem:var} and Theorems~\ref{thm:cadlag-plus} and~\ref{thm:cadlag-minus} will continue to hold
  if we replace all entries of $\var$ by $\var^+$ and all entries of $I$ by $J$.
\end{remark}

\begin{remark}\label{rem:cadlag}
  In this remark we allow the price path $\omega$ to take value zero.
  Redefine $\Omega$ as the set of all positive c\`adl\`ag functions $\omega:[0,1]\to[0,\infty)$,
  and consider the partition of $\Omega$ into the following three subsets:
  \begin{align*}
    A &:= \{\omega\in\Omega \st \inf\omega>0\},\\
    B &:= \{\omega\in\Omega \st \exists t\in[0,1]:\sign\omega=\III_{[0,t)}\},\\
    C &:= \Omega\setminus(A\cup B),
  \end{align*}
  where
  \begin{equation*}
    \sign u
    :=
    \begin{cases}
      1 & \text{if $u>0$}\\
      0 & \text{if $u=0$}\\
      -1 & \text{if $u<0$}.
    \end{cases}
  \end{equation*}
  In other words, $A$ is $\Omega$ as defined in the main part of this section,
  $B$ is the set of all $\omega\in\Omega$ that become zero at some point $t$ in time and then never recover,
  and $C$ is the set of all $\omega\in\Omega$ such that $\omega(t_1-)\wedge\omega(t_1)=0$ and $\omega(t_2)>0$ for some $t_1<t_2$.
  Theorems~\ref{thm:cadlag-plus} and~\ref{thm:cadlag-minus},
  as stated originally or as modified in the previous remark,
  describe the sign of $\UpProb$ for subsets of $A$.
  We can ignore the price paths in $C$:
  $\UpProb(C)=0$ and, therefore,
  for any $E\subseteq\Omega$,
  \begin{equation*}
    \UpProb(E)
    =
    \UpProb(E\cap(A\cup B)).
  \end{equation*}
  In the rest of this remark we will allow not only $f:[0,1]\to[-\infty,\infty)$
  in \eqref{eq:var-plus}, \eqref{eq:var-mod}, and \eqref{eq:var+mod},
  but also any $\omega\in\Omega$ (for the new definition of $\Omega$) in \eqref{eq:J-}--\eqref{eq:J+}.
  Theorem~\ref{thm:cadlag-plus} will continue to hold for $\omega\in A\cup B$
  if we replace $\var$ by $\var^+$ and $I$ by $J$
  (as shown by the same argument, given in Section~\ref{sec:proofs}).
  Equation~\eqref{eq:cadlag-minus} in Theorem~\ref{thm:cadlag-minus}
  can be rewritten as
  \begin{equation*}
    \UpProb
    \left\{
      \omega\in A
      \st
      J^-_{\omega}=\{t\},
      J^+_{\omega}=\emptyset
    \right\}
    =
    1.
  \end{equation*}
\end{remark}

\section{Positive price paths}
\label{sec:positive}

Let us now redefine $\Omega$ to be the set of all positive functions $\omega:[0,1]\to[0,\infty)$
satisfying $\inf\omega>0$
(without any continuity requirements).
The definitions of adapted processes, stopping times, etc., again stay as in Section~\ref{sec:continuous}.
Theorems~\ref{thm:cadlag-plus} and~\ref{thm:cadlag-minus} will still hold,
as shown by the same arguments in the next section.
Remark~\ref{rem:predictable} will still hold
with the same definitions of predictable and strongly predictable positive capital processes.
Remark~\ref{rem:cadlag} will still hold for $\Omega$
the set of all positive functions $\omega:[0,1]\to[0,\infty)$.

\section{Proofs of the theorems}
\label{sec:proofs}

The next result (Lemma~\ref{lem:prediction}) is applicable to all $\Omega$
considered in Sections~\ref{sec:continuous}--\ref{sec:positive}.
Fix a well-order $\preceq$ of $\Omega$, which exists by the Zermelo theorem
(one of the alternative forms of the Axiom of Choice; see, e.g., \cite[Theorem~5.1]{Jech:2003}).
Let $\omega^a$, where $\omega\in\Omega$ and $a\in[0,1]$,
be the $\preceq$-smallest element of $\Omega$ such that $\omega^a|_{[0,a]}=\omega|_{[0,a]}$.
Intuitively, using $\omega^a$ as the prediction at time $a$ for $\omega$ is an instance of Ockham's razor:
out of all hypotheses compatible with the available data $\omega|_{[0,a]}$ we choose the simplest one,
where simplicity is measured by the chosen well-order.

For any $\omega\in\Omega$ set
\begin{align}
  W_{\omega}
  &:=
  \left\{
    t\in[0,1]
    \st
    \forall t'\in(t,1]: \omega^{t'}\ne\omega^t
  \right\}
  \notag\\
  &=
  \left\{
    t\in[0,1]
    \st
    \forall t'\in(t,1]: \omega^{t'}\succ\omega^t
  \right\}
  \notag\\
  &=
  \left\{
    t\in[0,1]
    \st
    \forall t'\in(t,1]: \omega^t|_{[0,t']}\ne\omega|_{[0,t']}
  \right\}
  \label{eq:W}
\end{align}
(in particular, $1\in W_{\omega}$).
The following lemma says, intuitively, that short-term prediction of the future is usually possible.
\begin{lemma}\label{lem:prediction}
  \begin{enumerate}
  \item\label{it:small}
    The set $W_{\omega}$ is well-ordered by $\le$.
    (Therefore, each of its points is isolated on the right,
    which implies that $W_{\omega}$ is countable and nowhere dense.)
  \item\label{it:not-W}
    If $t\in[0,1]\setminus W_{\omega}$,
    there exists $t'>t$ such that $\omega^t|_{[t,t']}=\omega|_{[t,t']}$.
  \item\label{it:W}
    If $t\in[0,1)$,
    there exists $t'>t$ such that $\omega^{s}|_{[t,t']}=\omega|_{[t,t']}$
    for all $s\in(t,t')$.
  \end{enumerate}
\end{lemma}

Part~\ref{it:small} of Lemma~\ref{lem:prediction} says, informally, that the set $W_{\omega}$ is small.
Part~\ref{it:not-W} says that at each time point $t$
outside the small set $W_{\omega}$ the Ockham prediction system
that outputs $\omega^t$ as its prediction is correct (over some non-trivial time interval).
And part~\ref{it:W} says that even at time points $t$ in $W_{\omega}$
the Ockham prediction system becomes correct (in the same weak sense) immediately after time~$t$.

\begin{proof}
  Let us first check that $W_{\omega}$ is well-ordered by $\le$.
  Suppose there is an infinite strictly decreasing chain $t_1>t_2>\cdots$ of elements of $W_{\omega}$.
  Then we have $\omega^{t_1}\succ\omega^{t_2}\succ\cdots$,
  which contradicts $\preceq$ being a well-order.

  Each point $t\in W_{\omega}\setminus\{1\}$ is isolated on the right since $W_{\omega}\cap(t,t')=\emptyset$,
  where $t'$ is the successor of $t$.
  Therefore, $W_{\omega}$ is nowhere dense.
  To check that $W_{\omega}$ is countable,
  map each $t\in W_{\omega}\setminus\{1\}$ to a rational number in the interval $(t,t')$,
  where $t'$ is the successor of $t$;
  this mapping is an injection.

  As part~\ref{it:not-W} is obvious (and essentially asserted in~\eqref{eq:W}),
  let us check part~\ref{it:W}.
  Suppose $t\in[0,1)$.
  The set of all $\omega^s$, $s\in(t,1]$,
  has a smallest element $\omega^{t'}$, where $t'\in(t,1]$.
  It remains to notice that $\omega^s=\omega^{t'}$ for all $s\in(t,t')$.
\end{proof}

\begin{remark}
  It might be tempting to conjecture that, for any $t\in W_{\omega}\setminus\{1\}$,
  the function $s\mapsto\omega^s$ does not depend on $s\in(t,t')$,
  where $t'$ is the successor of $t$.
  While this statement is true for $\Omega=C[0,1]$,
  simple examples show that it is wrong in general:
  see Lemma~\ref{lem:counter-example} in Appendix~\ref{app:prediction}.
\end{remark}

\subsection{Proof of Theorem~\protect\ref{thm:continuous}}

For each pair of rational numbers $(a,b)$ such that $0<a<b<1$ fix a strictly positive weight $w_{a,b}>0$
such that $\sum_{a,b} w_{a,b} = 1$,
the sum being over all such pairs.
For each such pair $(a,b)$ we will define a positive capital process $\K^{a,b}$ such that $\K^{a,b}_0=1$
and $\K^{a,b}_b(\omega)=\infty$ when $\omega|_{[a,b]}=\omega^a|_{[a,b]}$ and $\omega|_{[a,b]}$ is not constant.
Let us check that the process
\begin{equation}\label{eq:K}
  \mathfrak{S}
  :=
  \sum_{a,b}
  w_{a,b}
  \K^{a,b}
\end{equation}
will then achieve our goal \eqref{eq:continuous}.

Let $\omega\in\Omega$ and $c$ be the largest $t\in[0,1]$ such that $\omega|_{[0,t]}$ is constant
(the supremum is attained by the continuity of $\omega$).
Assuming that $\omega$ is not constant, we have $c<1$.
Set $\omega^{c+}:=\omega^t$ for $t\in(c,c+\epsilon)$ for a sufficiently small $\epsilon>0$
(namely, such that $t\mapsto\omega^t$ is constant over the interval $(c,c+\epsilon)$;
such an $\epsilon$ exists by Lemma~\ref{lem:prediction}).
Choose $d\in(c,1)$ such that $\omega^d=\omega^{c+}$
(and, therefore, $\omega|_{(c,d]}=\omega^{c+}|_{(c,d]}$ and $\omega^t=\omega^{c+}$ for all $t\in(c,d]$).
Take rational $a,b\in(c,d)$ such that $a<b$ and $\omega|_{[a,b]}$ is not constant;
since $\K^{a,b}_b(\omega)=\infty$, \eqref{eq:K} gives $\mathfrak{S}_b(\omega)=\infty$;
and since $b$ can be arbitrarily close to $c$, we obtain \eqref{eq:continuous}.

It remains to construct such a positive capital process $\K^{a,b}$ for fixed $a$ and $b$.
From now until the end of this proof,
$\omega$ is a generic element of $\Omega$.
For each $n\in\{1,2,\ldots\}$,
let $\mathbb{D}_n:=\{k2^{-n}\st k\in\mathbb{Z}\}$
and define a sequence of stopping times $T^n_k$, $k=-1,0,1,2,\ldots$, inductively
by $T^n_{-1}:=a$,
\begin{align*}
  T^n_0(\omega)
  &:=
  \inf
  \left\{
    t\in[a,b]
    \st
    \omega(t)\in\mathbb{D}_n
  \right\},\\
  T^n_k(\omega)
  &:=
  \inf
  \left\{
    t\in[T^n_{k-1}(\omega),b]
    \st
    \omega(t)\in\mathbb{D}_n
    \And
    \omega(t)\ne\omega(T^n_{k-1})
  \right\},
  \enspace
  k=1,2,\ldots,
\end{align*}
where we set $\inf\emptyset:=b$.
For each $n=1,2,\ldots$,
define a simple capital process $\K^n$ as the capital process of the simple trading strategy
with the stopping times
\begin{equation*}
  \omega\in\Omega\mapsto\tau^n_k(\omega):=T^n_k(\omega)\wedge T^n_k(\omega^a),
  \quad
  k=0,1,\ldots,
\end{equation*}
the corresponding bets $h^n_k$ that are defined as
\begin{equation*}
  h^n_k(\omega)
  :=
  \begin{cases}
    2^{2n} \left(\omega^a(\tau^n_{k+1}(\omega^a))-\omega(\tau^n_k)\right)
      & \text{if $\omega^{\tau^n_k(\omega)}=\omega^a$ and $\tau^n_k(\omega)<b$}\\
    0 & \text{otherwise},
  \end{cases}
\end{equation*}
and an initial capital of 1.
Since the increments of this simple capital process never exceed~1 in absolute value
(and trading stops as soon as the prediction $\omega^a$ is falsified),
its initial capital of 1 ensures that it always stays positive.
The final value $\K^n_b(\omega)$ is $\Omega(2^{n})$
(to use Knuth's asymptotic notation)
unless $\omega|_{[a,b]}\ne\omega^a|_{[a,b]}$ or $\omega|_{[a,b]}$ is constant.
If we now set
\begin{equation*}
  \K^{a,b}
  :=
  \sum_{n=1}^{\infty}
  n^{-2}
  \K^{n},
\end{equation*}
we will obtain $\K^{a,b}_a<\infty$ and $\K^{a,b}_b(\omega)=\infty$ unless
$\omega|_{[a,b]}\ne\omega^a|_{[a,b]}$ or $\omega|_{[a,b]}$ is constant.
This completes the proof of Theorem~\ref{thm:continuous}.

\subsection{Proof of Theorem~\protect\ref{thm:cadlag}}
\label{subsec:thm-cadlag}

We will follow the proof of Proposition~2 in \cite{\CTV}
(that proposition considers measurable strategies,
but the assumption of measurability is not essential there).
Let us check the equivalent statement~\eqref{eq:e-var}.
If $c<\var^+(\log\omega)$,
we can find a partition $0=t_0<t_1<\cdots<t_n=1$ of $[0,1]$
such that
\begin{equation*}
  \sup
  \sum_{i=1}^n
  \left(
    \log\omega(t_i) - \log\omega(t_{i-1})
  \right)^+
  >
  c
\end{equation*}
(cf.\ \eqref{eq:var-plus}).
By investing all the available capital into $\omega$ at time $t_{i-1}$
whenever
$
  (
    \log\omega(t_i) - \log\omega(t_{i-1})
  )^+
  >
  0
$
(i.e., whenever $\omega(t_i)>\omega(t_{i-1})$),
the trader can turn 1 into at least $e^{c}$.
This proves the inequality $\le$ in \eqref{eq:e-var}.
And it is clear that this is the best the trader can do without risking bankruptcy.

For further (obvious) details, see the proof of Proposition~2 in \cite{\CTV}.

\subsection{Proof of Theorem~\protect\ref{thm:cadlag-plus}}
\label{subsec:thm-cadlag-plus}

The proof will use the fact that $\inf I^+_{\omega}\in I^+_{\omega}$ when $I^+_{\omega}\ne\emptyset$.
Notice that $I^+_{\omega}=J^+_{\omega}$,
where $J^+_{\omega}$ is defined in Remark~\ref{rem:var-plus}.

In this subsection we construct a positive capital process $\mathfrak{S}$ such that $\mathfrak{S}_0<\infty$
and $\mathfrak{S}_1=\infty$ whenever $I^+_{\omega}\ne\emptyset$;
moreover, it will satisfy \eqref{eq:cadlag-plus}.
Namely, we define $\mathfrak{S}$ via its representation~\eqref{eq:positive-capital}
with the components $\K^{G_n,c_n}(\omega)$, where $\omega$ is a generic element of $\Omega$,
defined as follows:
\begin{itemize}
\item
  $c_n=1/n^2$ (which ensures that the total initial capital $\sum_n 1/n^2$ is finite);
  $G_n$ will consist of stopping times denoted as $\tau^n_1,\tau^n_2,\ldots$
  and bets denoted as $h^n_1,h^n_2,\ldots$;
\item
  if $I^+_{\omega}=\emptyset$,
  set $\tau^n_1(\omega)=\tau^n_2(\omega)=\cdots=1$ and $h^n_1(\omega)=h^n_2(\omega)=\cdots=0$
  (intuitively, $G_n$ never bets,
  which makes this part of the definition non-anticipatory);
  in the rest of this definition we will assume that $I^+_{\omega}\ne\emptyset$ and, therefore, $\inf I^+_{\omega}<1$;
\item
  set $a:=\inf I^+_{\omega}$; we know that $a\in I^+_{\omega}$ and $a<1$;
\item
  in view of Lemma~\ref{lem:prediction},
  set $\omega^{a+}:=\omega^t$ for $t\in(a,a+\epsilon)$ for a sufficiently small $\epsilon$
  (such that $t\mapsto\omega^t$ does not depend on $t\in(a,a+\epsilon)$);
\item
  define
  \begin{equation*}
    c
    :=
    \inf
    \left\{
      t \in (a,(a+2^{-n})\wedge 1]
      \st
      \var^+\left(\log\omega^{a+}|_{[t,t+2^{-n}]}\right) \le n
    \right\}
  \end{equation*}
  (with $\inf\emptyset:=(a+2^{-n})\wedge 1$);
\item
  set $d:=(a+c)/2$ and define
  \begin{equation}\label{eq:tau-n-k-a}
    \tau^n_k(\omega^{a+})\in[d,(d+2^{-n})\wedge1]\cup\{1\}
  \end{equation}
  and $h^n_k(\omega^{a+})$, $k=1,2,\ldots$,
  in such a way that
  \begin{equation}\label{eq:h-n-k-a}
    h^n_k(\omega^{a+})
    \in
    \left[
      0,
      \frac{\K^{G_n,c_n}_{\tau^n_k(\omega^{a+})}(\omega^{a+})}{\omega^{a+}(\tau^n_k)}
    \right]
  \end{equation}
  (which is required implicitly by the definition of the positivity of $\mathfrak{S}$
  and is equivalent to the relative bets being in the range $[0,1]$)
  and
  \begin{equation}\label{eq:K-greater}
    \K^{G_n,c_n}_{(d+2^{-n})\wedge1}(\omega^{a+}) > e^n c_n;
  \end{equation}
  the latter can be done because of
  \begin{equation*}
    \var^+
    \left(
      \log\omega^{a+}|_{\left[d,d+2^{-n}\right]}
    \right)
    >
    n
  \end{equation*}
  and the fact that \eqref{eq:e-var} remains true
  if only positive simple capital processes are used as $\mathfrak{S}$
  in the definition \eqref{eq:upper-probability} of $\UpProb$
  (as can be seen from the proof in Subsection~\ref{subsec:thm-cadlag});
\item
  set
  \begin{equation}\label{eq:tau-n-k}
    \tau^n_k(\omega)
    :=
    \begin{cases}
      \tau^n_k(\omega^{a+}) & \text{if $\omega|_{[0,\tau^n_k(\omega^{a+})]}=\omega^{a+}|_{[0,\tau^n_k(\omega^{a+})]}$}\\
      1 & \text{otherwise}
    \end{cases}
  \end{equation}
  and
  \begin{equation*}
    h^n_k(\omega)
    :=
    \begin{cases}
      h^n_k(\omega^{a+}) & \text{if $\omega|_{[0,\tau^n_k(\omega^{a+})]}=\omega^{a+}|_{[0,\tau^n_k(\omega^{a+})]}$}\\
      0 & \text{otherwise}.
    \end{cases}
  \end{equation*}
\end{itemize}

Let us check \eqref{eq:cadlag-plus}.
Suppose the antecedent of \eqref{eq:cadlag-plus} holds for given $t\in[0,1]$ and $\omega\in\Omega$.
Using the notation introduced in the previous paragraph
(and suppressing the dependence on $\omega$ and $n$, as before),
we can see that $t>a$.
From some $n$ on we will have $d+2^{-n}<t$
and $\omega^s=\omega^{a+}$ for all $s\in(a,d+2^{-n})$,
and so the divergence of the series $\sum_n e^n/n^2$ implies that $\mathfrak{S}_t(\omega)=\infty$.

\begin{remark}\label{rem:predictable-proof}
  Let us check that the proof of Theorem~\ref{thm:cadlag-plus} can be modified
  to prove the stronger statement in Remark~\ref{rem:predictable}.
  Notice that, without loss of generality,
  we can replace the interval $[0,\ldots]$ in~\eqref{eq:h-n-k-a}
  by the two-element set $\{0,\ldots\}$ consisting of its end-points
  (see the proof in Subsection~\ref{subsec:thm-cadlag});
  this will ensure that the relative bets always satisfy $H^n_k(\omega)\in\{0,1\}$.
  In addition to the stopping times~\eqref{eq:tau-n-k-a} and bets~\eqref{eq:h-n-k-a},
  define stopping times $\sigma^n_k$ in such a way that:
  \begin{itemize}
  \item
    $\sigma^n_k(\omega^{a+})\in[\tau^n_k(\omega^{a+}),\tau^n_{k+1}(\omega^{a+})]$
    and
    $\sigma^n_k(\omega^{a+})\in(\tau^n_k(\omega^{a+}),\tau^n_{k+1}(\omega^{a+}))$
    when $\tau^n_k(\omega^{a+})<\tau^n_{k+1}(\omega^{a+})$;
  \item
    $\sigma^n_k(\omega^{a+})$ are so close to $\tau^n_k(\omega^{a+})$
    that \eqref{eq:K-greater} still holds
    when we replace the stopping times $\tau^n_k$ by $\sigma^n_k$
    in the definition of $G_n$
    (with the relative bets corresponding to~\eqref{eq:h-n-k-a} unchanged);
  \item
    for an arbitrary $\omega\in\Omega$,
    \begin{equation*}
      \sigma^n_k(\omega)
      :=
      \begin{cases}
        \sigma^n_k(\omega^{a+}) & \text{if $\omega|_{[0,\tau^n_k(\omega^{a+})]}=\omega^{a+}|_{[0,\tau^n_k(\omega^{a+})]}$}\\
        1 & \text{otherwise}
      \end{cases}
    \end{equation*}
    (cf.\ \eqref{eq:tau-n-k}).
  \end{itemize}
  After changing all $G_n$ in this way, we will obtain a positive capital process
  that is strongly predictable and still satisfies~\eqref{eq:cadlag-plus}.
\end{remark}

\subsection{Proof of Theorem~\protect\ref{thm:cadlag-minus}}
\label{subsec:thm-cadlag-minus}

We will be proving~\eqref{eq:cadlag-minus} for a fixed $t\in(0,1]$.
Fix a strictly increasing sequence $t_1,t_2,\ldots$ of numbers in the interval $(0,1)$ that converge to $t$,
$t_i\uparrow t$ as $i\to\infty$;
set $t_0:=0$.
Let $\Xi:=\{-1,1\}^{\infty}$.
For each sequence $\xi=(\xi_1,\xi_2,\ldots)\in\Xi$ define $\omega_{\xi}$
as the c\`adl\`ag function on $[0,t)$ that is constant on each of the intervals $[t_{i-1},t_i)$,
$i=1,2,\ldots$,
and satisfies $\omega_{\xi}(0):=1$ and
\begin{equation}\label{eq:omega-xi}
  \omega_{\xi}(t_i)
  :=
  \omega_{\xi}(t_{i-1})
  \left(
    1 + \frac{\xi_i}{i+1}
  \right),
  \quad
  i=1,2,\ldots\,.
\end{equation}

We will be particularly interested in $\omega_{\xi}$ such that $\lim_{i\to\infty}\omega_{\xi}(t_i)$ exists in $(0,\infty)$;
we can then extend $\omega_{\xi}$ to an element $\omega_{\xi\to}$ of $\Omega$ that is constant over $[t,1]$.
We will call such $\omega_{\xi}$ \emph{extendable};
for them $\omega_{\xi\to}$ exists and is an element of the set
$
  \{
    \omega\in\Omega
    \st
    I^-_{\omega}=\{t\},
    I^+_{\omega}=\emptyset
  \}
$
in~\eqref{eq:cadlag-minus}.

Let us check that no positive capital process $\mathfrak{S}$ grows by a factor of at least $1+\epsilon$,
where $\epsilon>0$ is a given constant,
on each extendable $\omega_{\xi}$.
Suppose, on the contrary, that a given $\mathfrak{S}$ satisfies $\mathfrak{S}_0=1$
and
\begin{equation}\label{eq:assumption}
  \mathfrak{S}_t(\omega_{\xi\to})\ge1+\epsilon
\end{equation}
for all extendable $\omega_{\xi}$.
Consider any representation of $\mathfrak{S}$ in the form~\eqref{eq:positive-capital}.

This proof uses methods of measure-theoretic probability;
our probability space is $\Xi$
equipped with the canonical filtration $(\mathcal{F}_i)$
and the power of the uniform probability measure on $\{-1,1\}$:
$\mathcal{F}_i$ consists of all subsets of $\Xi$ that are unions of cylinders
$\{(\xi_1,\xi_2,\ldots)\in\Xi\st\xi_1=c_1,\ldots,\xi_i=c_i\}$,
and the measure of each such cylinder is $2^{-i}$.
This is a discrete probability space without any measurability issues
(the simple idea of using such a ``poor'' probability space was used earlier in, e.g., \cite[Section~4.3]{Shafer/Vovk:2001}).

To simplify formulas we use the notation
$\omega_i:=\omega_{\xi}(t_i)$ and $K_i:=\mathfrak{S}_{t_i}(\omega)$
for any $\omega\in\Omega$ that agrees with $\omega_{\xi}$ over the interval $[0,t_i]$
(there is no dependence on such $\omega$,
and the dependence on $\xi$ is suppressed, as usual in measure-theoretic probability).
According to \eqref{eq:omega-xi}, $\omega_i$ is a martingale.
Let us check that
\begin{equation}\label{eq:K-transform}
  K_i = K_{i-1} + b_{i}(\omega_i-\omega_{i-1}),
\end{equation}
where $b_{i}$ is the total bet of all $G_n$ immediately before time $t_{i}$
after observing $\omega_{\xi}|_{[0,t_{i})}$
(the only issue in this check is convergence).
Formally,
\begin{equation}\label{eq:series}
  b_i
  :=
  \sum_{n=1}^{\infty}
  h^n_{k(n,i)},
\end{equation}
where
\begin{equation*}
  k(n,i)
  :=
  \max\{k\st\tau^n_k<t_i\}
\end{equation*}
and $(\tau^n_{k})$ and $(h^n_{k})$ are the stopping times and bets of $G_n$.
The series \eqref{eq:series} converges in $[0,\infty]$
as its terms are positive (to ensure the positivity of each $\K^{G_n,c_n}$).
Since the capital process $\mathfrak{S}$ is positive and $\omega$ can drop almost to $0$ at any time,
we have $b_{i}\in[0,K_{i-1}/\omega_{i-1}]$
(this follows from the analogous inclusions for the component simple capital processes).
We can see that \eqref{eq:K-transform} is indeed true.
Let us define $\alpha_{i}\in[0,1]$ by the condition
\begin{equation}\label{eq:alpha}
  b_{i}=\alpha_{i}K_{i-1}/\omega_{i-1}.
\end{equation}

Being a martingale transform of $\omega_i$, $K_i$ is also a martingale.
Combining~\eqref{eq:omega-xi}, \eqref{eq:K-transform}, and~\eqref{eq:alpha},
we can see that the recurrences for the two martingales are
\begin{align*}
  \omega_i
  &=
  \omega_{i-1}
  \left(
    1 + \frac{\xi_i}{i+1}
  \right),\\
  K_i
  &=
  K_{i-1}
  \left(
    1 + \alpha_{i} \frac{\xi_i}{i+1}
  \right),
  \quad
  i=1,2,\ldots\,.
\end{align*}
Let us check that $\log\omega_i$ converge in $\mathbb{R}$ for almost all $\xi$ as $i\to\infty$.
This follows from Taylor's formula
\begin{equation*}
  \log\omega_i - \log\omega_{i-1}
  =
  \frac{\xi_i}{i+1}
  -
  \frac12
  \frac{1}{(1+\theta_i\xi_i/(i+1))^2}
  \frac{1}{(i+1)^2}
\end{equation*}
(where $\theta_i\in[0,1]$),
the almost sure convergence of $\sum_i \xi_i/(i+1)$
(which follows from Kolmogorov's two series theorem),
and the convergence of $\sum_i (i+1)^{-2}$.
We can see that $\omega_{\xi}$ is almost surely extendable.
In the same way we can demonstrate the convergence of $\log K_i$ in $\mathbb{R}$,
but we will not need it.

By Fatou's lemma, we have,
for any $\xi\in\Xi$ with extendable $\omega_{\xi}$,
\begin{multline*}
  \mathfrak{S}_t
  (\omega_{\xi\to})
  =
  \sum_{n=1}^{\infty}
  \K^{G_n,c_n}_t
  (\omega_{\xi\to})
  =
  \sum_{n=1}^{\infty}
  \lim_{i\to\infty}
  \K^{G_n,c_n}_{t_i}
  (\omega_{\xi})\\
  \le
  \liminf_{i\to\infty}
  \sum_{n=1}^{\infty}
  \K^{G_n,c_n}_{t_i}
  (\omega_{\xi})
  =
  \liminf_{i\to\infty}
  K_i.
\end{multline*}
Another application of Fatou's lemma and the fact that almost all $\omega_{\xi}$ are extendable
show that the chain
\begin{equation*}
  \Expect
  \mathfrak{S}_t
  (\omega_{\xi\to})
  \le
  \Expect
  \liminf_{i\to\infty}
  K_i
  \le
  \liminf_{i\to\infty}
  \Expect
  K_i
  =
  1
\end{equation*}
is well defined and correct.
This contradicts our assumption~\eqref{eq:assumption}.

\section{Conclusion}

This paper shows that some assumptions of regularity
(apart from being non-anticipative, such as universal measurability)
should be imposed on continuous-time trading strategies
even in game-theoretic probability.
This is not a serious problem in applications
since only computable trading strategies can be of practical interest,
and computable trading strategies will be measurable
under any reasonable computational model.

There are many interesting directions of further research, such as:
\begin{itemize}
\item
  Is it possible to extend Theorems~\ref{thm:continuous}, \ref{thm:cadlag-plus}, and~\ref{thm:cadlag-minus}
  to the case where only the most recent past is known to the trader,
  as in \cite[p.~vii and Section~7.3]{Hardin/Taylor:2013} and \cite[Section~5]{Hardin/Taylor:2008}?
\item
  Is it possible to extend Theorems~\ref{thm:continuous}, \ref{thm:cadlag-plus}, and~\ref{thm:cadlag-minus}
  to the case of the trader without a synchronized watch
  (see \cite[Section~7.7]{Hardin/Taylor:2013} or \cite{Bajpai/Velleman:2016})?
\item
  The construction in Subsection~\ref{subsec:thm-cadlag-minus}
  produces a set $E\subseteq\Omega$ (consisting of all $\omega_{\xi\to}$ for $\xi\in\Xi$ with extendable $\omega_{\xi}$)
  that satisfies both $\UpProb(E)=1$ and $\var(\log\omega)=\infty$ for all $\omega\in E$.
  However, $\vi(\log\omega)=1$ for all $\omega\in E$
  (see, e.g., \cite[Subsection~4.2]{\CTIV} for the definition of variation index $\vi$).
  Do $E\subseteq\Omega$ with $\UpProb(E)=1$ and $\inf_{\omega\in E}\vi(\log\omega)>1$ exist?
\end{itemize}

\subsection*{Acknowledgements}

I am grateful to the participants in the workshop
``Pathwise methods, functional calculus and applications in mathematical finance''
(Vienna, 4--6 April 2016) for their comments.
Thanks to Yuri Gurevich for numerous discussions on a wide range of topics
(as a result of which I discovered Hardin and Taylor's work on hat puzzles)
and to Yuri Kalnishkan for a useful discussion on the topic of this paper.
The anonymous reviewers' thoughtful comments made me change the presentation and emphasis;
they also prompted me to include Appendix~A.

  This work has been supported by the Air Force Office of Scientific Research
  (project ``Semantic Completions'').

\appendix
\section{Fine structure of prediction with the Axiom of Choice}
\label{app:prediction}

In the introductory part of Section~\ref{sec:proofs} we introduced the Ockham prediction systems
(parameterised by well-orders on $\Omega$)
but limited ourselves to results required in the rest of that section.
In this appendix we extend these prediction systems and study them more systematically;
the new notions and results of this appendix are not needed in the main part of the paper.

Remember that $\preceq$ is a well-order on $\Omega$,
which in this appendix can be any of the $\Omega$ considered in the main part of the paper
(unless $\Omega$ is explicitly pointed out).
The \emph{Ockham prediction system} is:
\begin{itemize}
\item
  given $\omega|_{[0,t)}$, the prediction $\omega^{t-}$ for the rest of $\omega$
  is defined as the $\preceq$-smallest $\omega'$ such that $\omega'|_{[0,t)}=\omega|_{[0,t)}$;
\item
  given $\omega|_{[0,t]}$, the prediction for the rest of $\omega$ is $\omega^t$,
  as defined in Section~\ref{sec:proofs}:
  $\omega^t$ is the $\preceq$-smallest element $\omega'$ satisfying $\omega'|_{[0,t]}=\omega|_{[0,t]}$;
\item
  the revised prediction $\omega^{t+}$ at the time $t\in[0,1)$ is $\omega^{t'}$
  for any $t'\in(t,1]$ such that the function $s\mapsto\omega^s$ is constant over $(t,t']$.
\end{itemize}
The existence of $\omega^{t+}$ was shown in part~\ref{it:W} of Lemma~\ref{lem:prediction}
(and was already used in the proof of Theorem~\ref{thm:cadlag-plus} in Subsection~\ref{subsec:thm-cadlag-plus}).
By definition, $\omega^{1+}$ is undefined, but notice that $\omega^t$ and $\omega^{t-}$
are defined for all $t\in[0,1]$
(in particular, $\omega^{0-}$ is the $\preceq$-smallest element of $\Omega$).

We have the following three dichotomies for time points $t\in[0,1]$
for the purpose of short-term prediction of a given $\omega\in\Omega$:
\begin{itemize}
\item
  $t$ is \emph{past-successful} (for $\omega$)
  if there exists $t'<t$ such that $\omega^{t'}=\omega^{t-}$;
  in particular, $0$ is not past-successful;
\item
  $t$ is \emph{present-successful} if $\omega^{t-}=\omega^{t}$;
\item
  $t$ is \emph{future-successful} if $\omega^{t}=\omega^{t+}$;
  in particular, $1$ is not future-successful.
\end{itemize}
This gives us a partition of all $t\in[0,1]$ into $2^3=8$ classes.
We say that $t$ is \emph{$(-,0,+)$-successful}
(for the short-term prediction of $\omega$ using the Ockham prediction system)
if $t$ is simultaneously past-successful, present-successful, and future-successful;
this is the highest degree of success.
More generally, we will include $-$ (respectively, $0$, $+$) in the designation of the class of $t$
if and only if $t$ is past- (respectively, present-, future-) successful.
The $t$ that are $()$-successful are not successful at all:
they are not past-successful, not present-successful, and not future-successful.
We will use notation such as $C^{(-,0,+)}_{\omega}$ and $C^{()}_{\omega}$
for denoting the set of $t$ of the class indicated as the superscript.
For example, $C^{(-,0)}_{\omega}$ is the class of $t\in[0,1]$ that are past- and present-successful
but not future-successful for $\omega$.

The definition \eqref{eq:W} can be expressed in our new terminology
by saying that $W_{\omega}$ are the $t\in[0,1]$ that are not future-successful
(which agrees with $1\in W_{\omega}$);
in our new notation,
\begin{equation}\label{eq:W-as-union}
  W_{\omega}
  =
  C^{(-,0)}_{\omega} \cup C^{(-)}_{\omega} \cup C^{(0)}_{\omega} \cup C^{()}_{\omega}.
\end{equation}
Let us modify the definition \eqref{eq:W} by setting
\begin{equation*}
  F_{\omega}
  :=
  \left\{
    t\in[0,1]
    \st
    \left(
      \forall t'\in(t,1]: \omega^{t'}\ne\omega^t
    \right)
    \text{ or }
    \left(
      \forall t'\in[0,t): \omega^{t'}\ne\omega^t
    \right)
  \right\};
\end{equation*}
in particular, $\{0,1\}\subseteq F_{\omega}$.
This is the set of times $t\in[0,1]$ when the Ockham prediction system fails in the weakest possible sense
that can be expressed via our three dichotomies;
namely, $F_{\omega}$ is the set $[0,1]\setminus C^{(-,0,+)}$ of $t\in[0,1]$ that are not $(-,0,+)$-successful.
If $t\in[0,1]\setminus F_{\omega}$,
the Ockham prediction system correctly predicts $\omega|_{[t_1,t_2]}$ already at time $t_1$,
where $(t_1,t_2)\ni t$ is a neighbourhood of $t$.

It is always true that $W_{\omega}\subseteq F_{\omega}$, and even the set $F_{\omega}$ is still small:

\begin{lemma}\label{lem:F-small}
  The set $F_{\omega}$ is well-ordered by $\le$.
  (Therefore, each of its points is isolated on the right,
  which implies that $F_{\omega}$ is countable and nowhere dense.)
\end{lemma}

\begin{proof}
  We can modify the argument in the proof of Lemma~\ref{lem:prediction}, part~\ref{it:small}.
  Suppose there is an infinite strictly decreasing chain $t_1>t_2>\cdots$ of elements of $F_{\omega}$.
  Then $\omega^{t_1}\succeq\omega^{t_2}\succeq\cdots$,
  where $=$ is now possible.
  However, if $\omega^{t_i}=\omega^{t_{i+1}}$
  (i.e., $t_{i+1}$ is future-successful and $t_{i}$ is past- and present-successful for $\omega$)
  for $i$ in $\{2,3,\ldots\}$,
  then by the definition of $F_{\omega}$ we have $\omega^{t_{i-1}}\succ\omega^{t_{i}}$ and $\omega^{t_{i+1}}\succ\omega^{t_{i+2}}$.
  Therefore, if we remove the duplicates
  from the chain $\omega^{t_1}\succeq\omega^{t_2}\succeq\cdots$
  (i.e., replace each adjacent pair $\omega^{t_i}=\omega^{t_{i+1}}$ by just $\omega^{t_i}$),
  we will still have an infinite chain with $\succ$ in place of $\succeq$.
  This contradicts $\preceq$ being a well-order.
\end{proof}

Part~\ref{it:small} of Lemma~\ref{lem:prediction} is a special case of Lemma~\ref{lem:F-small},
since any subset of a well-ordered set is well-ordered.
We can also see that each of the eight classes apart from $C^{(-,0,+)}_{\omega}$ is well-ordered.  

The following lemma shows that $F_{\omega}$ splits $[0,1]$ into intervals of constancy of $t\mapsto\omega^t$.

\begin{lemma}\label{lem:constant}
  For any $t\in F_{\omega}\setminus\{1\}$,
  the function $s\mapsto\omega^s$ is constant on the interval $(t,t')$,
  where $t'$ is the successor of $t$ in $F_{\omega}$.
\end{lemma}

\begin{proof}
  Take any $t\in F_{\omega}\setminus\{1\}$, its successor $t'\in F_{\omega}$, and any $t''\in(t,t')$.
  Set
  \begin{align}
    t_1
    &:=
    \inf
    \left\{
      s\in(t,t')\st\omega^s=\omega^{t''}
    \right\},
    \label{eq:inf}\\
    t_2
    &:=
    \sup
    \left\{
      s\in(t,t')\st\omega^s=\omega^{t''}
    \right\}.
    \notag
  \end{align}
  It suffices to show that $t_1=t$ and $t_2=t'$.
  Suppose, e.g., that $t_1>t$.
  This implies $t_1\notin F_{\omega}$.
  By the definition of $F_{\omega}$, there is a neighbourhood of $t_1$ in which $s\mapsto\omega^s$ is constant
  and, therefore, $\omega^s=\omega^{t''}$;
  this, however, contradicts the definition \eqref{eq:inf} of $t_1$.
\end{proof}

Let us check that the analogue of Lemma~\ref{lem:constant} still holds for $W_{\omega}$ in place of $F_{\omega}$
if $\Omega=C[0,1]$ but fails in general.

\begin{lemma}\label{lem:counter-example}
  If $\Omega=C[0,1]$,
  for any $t\in W_{\omega}\setminus\{1\}$,
  the function $s\mapsto\omega^s$ is constant on the interval $(t,t')$,
  where $t'$ is the successor of $t$ in $W_{\omega}$.
  Otherwise,
  there are a well-order on $\Omega$, $\omega\in\Omega$, and $t\in W_{\omega}\cap(0,1)$
  such that the successor $t'$ of $t$ in $W_{\omega}$ is in $(0,1)$
  and the function $s\mapsto\omega^s$ is not constant over the interval $(t,t')$,
\end{lemma}

\begin{proof}
  Suppose $\Omega=C[0,1]$, $t\in W_{\omega}\setminus\{1\}$,
  and $t'$ is the successor of $t$ in $W_{\omega}$.
  Let $t^*\le t'$ be the successor of $t$ in $F_{\omega}$.
  By Lemma~\ref{lem:constant},
  $\omega^{t+}(s)=\omega(s)$ for each $s\in(t,t^*)$,
  and so our continuity assumption implies that $\omega^{t+}(s)=\omega(s)$ for each $s\in(t,t^*]$.
  This shows that $t'=t^*$,
  and so $s\mapsto\omega^s$ is constant on $(t,t')=(t,t^*)$.

  Now suppose $\Omega\ne C[0,1]$.
  In the remaining proofs we will freely use the fact that the sum of two well-orders is again a well-order
  \cite[Lemma 3.5(2)]{Rosenstein:1982},
  where the \emph{sum} $\preceq$ of orders $\preceq_1$ and $\preceq_2$ on two disjoint sets $X_1$ and $X_2$,
  respectively,
  is defined by
  \begin{equation*}
    x \preceq x'
    \Longleftrightarrow
    \begin{cases}
      x \preceq_1 x' & \text{if $x\in X_1$ and $x'\in X_1$}\\
      x \preceq_2 x' & \text{if $x\in X_2$ and $x'\in X_2$}\\
      \textbf{true} & \text{if $x\in X_1$ and $x'\in X_2$}\\
      \textbf{false} & \text{if $x\in X_2$ and $x'\in X_1$}\\
    \end{cases}
  \end{equation*}
  for any $x,x'\in X_1\cup X_2$ \cite[Definition 1.29]{Rosenstein:1982}.
  This implies that any well-order on part of $\Omega$ can be extended to a well-order on the whole of $\Omega$.

  Consider a well-order on $\Omega$ that starts from the c\`adl\`ag functions defined by
  \begin{align*}
    \omega_1
    &:=
    1,\\
    \omega_2(t)
    &:=
    \begin{cases}
      1 & \text{if $t\in[0,1/4]$}\\
      t+3/4 & \text{otherwise},
    \end{cases}\\
    \omega_3(t)
    &:=
    \begin{cases}
      1 & \text{if $t\in[0,1/4]$}\\
      t+3/4 & \text{if $t\in[1/4,1/2)$}\\
      1 & \text{otherwise},
    \end{cases}\\
    \omega_4(t)
    &:=
    \begin{cases}
      1 & \text{if $t\in[0,1/4]$}\\
      t+3/4 & \text{if $t\in[1/4,1/2)$}\\
      1 & \text{if $t\in[1/2,3/4]$}\\
      t+1/4 & \text{otherwise}
    \end{cases}
  \end{align*}
  (in this order).
  For $\omega:=\omega_4$, we have $W_{\omega}=\{1/4,3/4\}$,
  $3/4$ is the successor of $1/4$ in $W_{\omega}$,
  and the function $s\mapsto\omega^s$ is not constant over the interval $(1/4,3/4)$
  (it changes its value at $s=1/2$ from $\omega_2$ to $\omega_3$).
\end{proof}

The following result (easy but tiresome) shows that each of the eight classes may be non-empty,
apart from the case $\Omega=C[0,1]$, where there are six potentially non-empty classes
(two of them being subsets of $\{0\}$).
In particular, it implies that $W_{\omega}\ne F_{\omega}$ is possible
and, moreover, $F_{\omega}\setminus W_{\omega}$ may contain any $c\in(0,1)$
(which is also clear from the proof of Lemma~\ref{lem:counter-example}).

\begin{lemma}\label{lem:structure}
  The class $C^{(-,0,+)}_{\omega}$ is never empty
  (and is large in the sense of Lemma~\ref{lem:F-small}).
  Let $c\in(0,1)$.
  \begin{itemize}
  \item
    For $\Omega=C[0,1]$ of Section~\ref{sec:continuous},
    the classes $C^{(-,+)}_{\omega}$ and $C^{(-)}_{\omega}$ are always empty;
    the classes $C^{(+)}_{\omega}$ and $C^{()}_{\omega}$ are subsets of $\{0\}$
    (and can be equal both to $\{0\}$ and to $\emptyset$, depending on $\omega$ and the well-order $\preceq$);
    for each of the remaining three classes 
    (those containing $0$ in the superscript but different from $C^{(-,0,+)}_{\omega}$)
    there exist $\omega\in\Omega$ and $\preceq$ such that the class coincides with~$\{c\}$.
  \item
    For all other $\Omega$ considered in the main part of the paper,
    for each of the eight classes apart from $C^{(-,0,+)}_{\omega}$
    there exist $\omega\in\Omega$ and $\preceq$ such that the class coincides with~$\{c\}$.
  \end{itemize}
\end{lemma}

\begin{proof}
  First we consider the case $\Omega=C[0,1]$.
  The statement about the classes not containing $0$ in the superscript follows from any $t\in(0,1]$
  being present-successful for a continuous $\omega$.
  The classes $C^{(-,+)}_{\omega}$ and $C^{(-)}_{\omega}$ are empty
  since $0$ is never past-successful.
  If the three $\preceq$-smallest elements are $\omega_1:=1$, $\omega_2:=2$,
  and $\omega_3(t):=2+t$ (in this order),
  we have $C^{(+)}_{\omega_2}=\{0\}$, $C^{()}_{\omega_2}=\emptyset$,
  $C^{(+)}_{\omega_3}=\emptyset$, and $C^{()}_{\omega_3}=\{0\}$.

  It remains to consider, for this $\Omega$, the classes containing $0$ in the superscript.
  Let the well-order on $\Omega$ start from
  \begin{align*}
    \omega_1 &:= 1,\\
    \omega_2(t)
    &:=
    \begin{cases}
      1 & \text{if $t<c$}\\
      1+t-c & \text{otherwise},
    \end{cases}
  \end{align*}
  i.e., $\omega_1\prec\omega_2\prec\cdots$.
  We have $C^{(-,0)}_{\omega_2}=\{c\}$.
    
  Next consider a well-order on $\Omega$ that starts from
  \begin{align}
    \omega_1 &:= 2,
    \notag\\
    \omega_n(t)
    &:=
    \begin{cases}
      1+t & \text{if $t<c(1-1/n)$}\\
      1+c(1-1/n) & \text{otherwise},
    \end{cases}
    \quad
    n=2,3,\ldots,
    \label{eq:well-order-2}\\
    \omega_{\boldsymbol{\omega}}(t)
    &:=
    \begin{cases}
      1+t & \text{if $t<c$}\\
      1+c & \text{otherwise},
    \end{cases}
    \label{eq:well-order-3}\\
    \omega_{\boldsymbol{\omega}+1}(t)
    &:=
    1+t
    \notag
  \end{align}
  (in this order, with the boldface $\boldsymbol{\omega}$ standing for the first infinite ordinal).
  We have $C^{(0,+)}_{\omega_{\boldsymbol{\omega}}}=\{c\}$
  and $C^{(0)}_{\omega_{\boldsymbol{\omega}+1}}=\{c\}$.
  This completes the proof for $\Omega=C[0,1]$.

  Now let $\Omega\ne C[0,1]$ be any of the other $\Omega$ considered in the main part of the paper.
  Since $C[0,1]\subseteq\Omega$,
  we have already shown that $C^{(-,0)}_{\omega}$, $C^{(0,+)}_{\omega}$, and $C^{(0)}_{\omega}$
  can be equal to $\{c\}$.

  Consider a well-order on $\Omega$ that starts from
  \begin{align*}
    \omega_1 &:= 1,\\
    \omega_2(t)
    &:=
    \begin{cases}
      1 & \text{if $t<c$}\\
      2 & \text{otherwise},
    \end{cases}\\
    \omega_3(t)
    &:=
    \begin{cases}
      1 & \text{if $t<c$}\\
      2+t-c & \text{otherwise}
    \end{cases}
  \end{align*}
  (in this order).
  We have $C^{(-,+)}_{\omega_2}=\{c\}$ and $C^{(-)}_{\omega_3}=\{c\}$.

  Finally consider a well-order on $\Omega$ that starts from $\omega_1:=1$,
  \eqref{eq:well-order-2}--\eqref{eq:well-order-3},
  and
  \begin{align*}
    \omega_{\boldsymbol{\omega}+1}(t)
    &:=
    \begin{cases}
      1+t & \text{if $t<c$}\\
      2 & \text{otherwise},
    \end{cases}\\
    \omega_{\boldsymbol{\omega}+2}(t)
    &:=
    \begin{cases}
      1+t & \text{if $t<c$}\\
      2+t-c & \text{otherwise}
    \end{cases}
  \end{align*}
  (in this order).
  We have $C^{(+)}_{\omega_{\boldsymbol{\omega}+1}}=\{c\}$
  and $C^{()}_{\omega_{\boldsymbol{\omega}+2}}=\{c\}$.
\end{proof}

Lemma~\ref{lem:structure} shows that the number of different unions
(such as \eqref{eq:W-as-union})
that can be formed from the classes $C_{\omega}^{(\cdots)}$ is very large
(namely, $2^8=256$),
and many of these are potentially interesting.
For each of the unions we can ask what sets in $[0,1]$ can be represented as such a union for different $\omega$.
We will answer this question only for the union \eqref{eq:W-as-union},
which plays the most important role in this paper.

\begin{lemma}\label{lem:W}
  For any well-ordered set $W\subseteq[0,1]$ containing $1$,
  there exist $\omega\in\Omega$ and a well-oder $\preceq$
  such that $W_{\omega}=W$.
\end{lemma}

\begin{proof}
  It suffices to consider $\Omega=C[0,1]$,
  which will imply the analogous statement for any other $\Omega$ considered in this paper.
  Let $\alpha$ be the ordinal that is isomorphic to $W$ \cite[Theorem~2.12]{Jech:2003},
  and let $\beta\in\alpha\mapsto w_{\beta}\in W$ be the unique isomorphism \cite[Corollary~2.6]{Jech:2003}
  between $\alpha$ and $W$.
  Let $\preceq$ be a well-order that starts from $(\omega_{\beta})_{\beta\in\alpha}$,
  in the usual order of $\beta\in\alpha$,
  such that
  \begin{equation*}
    \omega_{\beta}(t)
    :=
    \begin{cases}
      t & \text{if $t<w_{\beta}$}\\
      w_{\beta} & \text{otherwise}.
    \end{cases}
  \end{equation*}
  It remains to set $\omega(t):=t$ for all $t\in[0,1]$.
\end{proof}

An analogue of Lemma~\ref{lem:W} (however, with $\supseteq$ in place of $=$)
is contained in Theorem~3.5 of \cite{Hardin/Taylor:2008}.
\end{document}